\documentclass[letterpaper, 10 pt, conference]{IEEEtran}
\IEEEoverridecommandlockouts

\usepackage{amsmath,amssymb, amsthm, mathtools}
\usepackage{algorithm}
\usepackage[noend]{algpseudocode}
\usepackage{tikz}
\usetikzlibrary{%
  arrows.meta, positioning, calc, shapes.geometric, fit, backgrounds,
  patterns, decorations.pathreplacing, tikzmark}
\usepackage{booktabs}
\usepackage{multirow}
\usepackage{xcolor}
\usepackage{hyperref}
\usepackage{cleveref}

\usepackage[normalem]{ulem}

\newtheorem{assumption}{Assumption}
\newtheorem{definition}{Definition}
\newtheorem{theorem}{Theorem}
\newtheorem{lemma}{Lemma}
\newtheorem{corollary}{Corollary}
\newtheorem{proposition}{Proposition}
\newtheorem{remark}{Remark}

\crefname{assumption}{Assumption}{Assumptions}
\Crefname{assumption}{Assumption}{Assumptions}
\crefrangelabelformat{assumption}{#3#1#4--#5#2#6}
\crefname{equation}{}{}
\Crefname{equation}{Equation}{Equations}

\DeclareMathOperator{\Pre}{Pre}
\DeclareMathOperator{\Bas}{Bas}

\newcommand{\Down}{\mathord{\downarrow}}
\newcommand{\N}{\mathbb{N}}
\newcommand{\R}{\mathbb{R}}

\definecolor{safeblue}{RGB}{70,130,180}
\definecolor{safefill}{RGB}{220,235,250}
\definecolor{dangerfill}{RGB}{255,235,235}
\definecolor{accentgreen}{RGB}{40,160,100}
\definecolor{accentorange}{RGB}{230,140,30}
\newif\ifshowdeletions
\showdeletionsfalse 
\newif\ifshowoldtable
\showoldtablefalse 
\newif\ifshowtimingbreakdown
\showtimingbreakdownfalse 
\newcommand{\added}[1]{\textcolor{black}{#1}}
\newcommand{\deleted}[1]{\ifshowdeletions\textcolor{blue}{\sout{#1}}\fi}
\newcommand{\replaced}[2]{\deleted{#1}\added{#2}}
\let\change\added

\makeatletter
\renewcommand{\theALG@line}{\thealgorithm.\arabic{ALG@line}}
\renewcommand{\theHALG@line}{\thealgorithm.\arabic{ALG@line}}
\makeatother

\title{Real-Time Synthesis of Robust Controlled Invariant Sets for Monotone Systems}
\author{\IEEEauthorblockN{Yasin Sonmez, Mahmoud Khaled, Majid Zamani, Murat Arcak}%
\thanks{Y. Sonmez and M. Arcak are with the 
University of California, Berkeley, USA (\{yasin\_sonmez,arcak\}@berkeley.edu); M. Zamani is with the 
University of Colorado Boulder, USA (majid.zamani@colorado.edu); M. Khaled is with the 
Ludwig-Maximilian University, Germany (m.mahmoud@lmu.de). This work is supported in part by the NSF grant CNS-2111688.}}
\begin{document}
\maketitle

\begin{abstract}
Safety-critical control of autonomous systems requires formal safety certificates, such as controlled invariant sets, that must be computed online as 
conditions change. 
Although standard synthesis algorithms scale poorly with state dimension,
monotone dynamical systems with lower-closed safety specifications allow for accelerated computation of controlled invariant sets.
In particular,
 \emph{lazy} fixed-point algorithms exploit monotonicity and track only the antichain basis of the set. However, membership tests and redundancy checks against an evolving basis remain major bottlenecks.
We introduce a threshold-function reformulation in which a lower-closed set on a $d$-dimensional grid 
is represented by its column heights along a designated axis.
This reformulates the greatest-fixed-point iteration as independent one-dimensional binary searches, one per grid column, yielding an embarrassingly parallel iteration with asymptotically lower computational complexity than the lazy fixed-point algorithm.
Experiments synthesize invariant sets on 3D grids with $10^{9}$ cells in under 50\,ms and $10^{14}$ cells in under two minutes.
\replaced{We further use these sets as online safety filters inside a real-time model predictive controller example.}{We further demonstrate online re-synthesis in a safety-informed model predictive controller example.}
\end{abstract}

\section{Introduction}\label{sec:intro}

Controlled invariant sets characterize safe regions in which trajectories can be kept for all future time through appropriate control actions.
Symbolic abstractions \cite{tabuada2009book} locate controlled invariant sets through finite-state fixed-point computations, but their cost grows sharply with state dimension.

\replaced{To address the scalability barrier, compositional~\cite{zamani2018compositional} and data-driven abstractions~\cite{devonport2021symbolic} have been pursued. Another}{One approach to overcome the scalability barrier} is to exploit monotone dynamics:
When the dynamics are monotone with respect to a partial order and the safe set is lower-closed with respect to the same order, the maximal robust controlled invariant set is also lower-closed and is uniquely determined by its \emph{antichain basis}~\cite{saoud2019efficient}.
On an $N^d$ grid, the basis size is $O(N^{d-1})$, one order of magnitude smaller than the underlying grid.
The \emph{lazy} fixed-point algorithms of~\cite{saoud2019efficient,ivanova2020lazy,ivanova2022lazy} 
track only the basis and update the antichain boundary between iterations.
Two bottlenecks remain despite this reduction:
\begin{enumerate}
  \item \textbf{Quadratic per-iteration cost.} Each iteration performs $O(N^{d-1})$ safety checks, and each check scans the basis to test membership in its lower closure, giving $O(N^{2(d-1)})$ work per iteration.
  \item \textbf{Sequential neighbor generation.} After unsafe elements are removed, new candidate neighbors must be added to the basis, each requiring a redundancy check---an inherently sequential step that limits parallelism.
\end{enumerate}

Instead of representing the boundary of the lower-closed set with 
basis points, which must be iteratively updated and searched, in this paper we store this surface as a threshold function along a designated axis.
We then reformulate the greatest-fixed-point (GFP) iteration to operate \emph{solely} on this threshold function, eliminating basis tracking and sequential neighbor generation.
The contributions are:
\begin{itemize}
  \item We encode lower-closed sets by a threshold function with $O(N^{d-1})$ storage and $O(1)$ membership queries. 
  Parallel column updates yield per-iteration work $O(N^{d-1}\log N)$ without basis enumeration or neighbor checks.
  \item We prove one-step equivalence with the predecessor operator and equivalence with the lazy basis formulation.
  \item \deleted{We show experimentally that the threshold iteration yields speedups of up to $21{,}924{\times}$ over basis-based lazy synthesis and supports real-time, closed-loop re-synthesis in 39--52\,ms on grids with over $10^{10}$ cells and in under two minutes on grids with $10^{14}$ cells.}\added{We show experimentally that the threshold iteration yields speedups on the order of thousands over basis-based lazy synthesis. The $10^{14}$-cell scaling benchmarks complete in under two minutes, while the separate $10^{10}$-cell closed-loop experiment completes in approximately 62--77\,ms.}
\end{itemize}

\added{Existing synthesis methods that leverage structural properties include sparse and compositional constructions for monotone systems and directed specifications~\cite{kim2017symbolic}, flexible abstraction-based pipelines~\cite{kim2019flexible}, lazy exploration of exposed boundaries using antichain representations~\cite{ivanova2022lazy}, and ordered reductions for robust controlled invariance~\cite{saoud2024characterization}. Height encodings of staircase sets and discrete Pareto frontiers are well-established. Our contribution lies not in the height encoding, but in its use as the sole representation in robust invariant-set GFP iteration, together with the column-prefix result, binary-search update, equivalence proof, and parallel implementation.}

\deleted{Abstraction tools for general (non-monotone) systems include \textsc{SCOTS}~\cite{rungger2016scots}, \textsc{ROCS}~\cite{li2018rocs}, and \textsc{Mascot}~\cite{hsu2018multi}. \textsc{pFaces}~\cite{khaled2019pfaces} provides a heterogeneous parallel back-end and is the foundation upon which
our implementation is built.}

\section{Preliminaries}\label{sec:prelim}

\subsection{Monotone systems and lower sets}

We study continuous-time systems of the form
\begin{equation}\label{eq:cts}
  \dot \xi = f(\xi,u,\vartheta),
\end{equation}
where $\xi \in \R^d$ is the continuous state, $u \in \R^m$ is the control input, and $\vartheta$ collects exogenous effects such as disturbances and uncertain parameters.
We obtain a finite model from~\eqref{eq:cts} using  sampled-data symbolic abstraction \cite{tabuada2009book,rungger2016scots,li2018rocs}.
Given a sampling period, state grid, and  finite input alphabet, we overestimate one-step reachable sets under all admissible exogenous signals and 
obtain a finite transition system,
\begin{equation}\label{eq:sys}
  x^+ = F(x,u,\delta).
\end{equation}
Here, $x$ belongs to a finite grid $X$, $u$ belongs to a finite admissible control set $\mathcal U$, and $\delta$ belongs to a finite adversarial set $\Delta$.
\replaced{Here, $\Delta$ is the abstraction-level uncertainty set: its elements index the successor cells induced by the admissible continuous exogenous signals $\vartheta$, uncertain parameters, and discretization over-approximation over one sampling interval.}{This set is finite at the abstraction level and conservatively indexes all abstract successor cells induced by every admissible continuous exogenous signal $\vartheta$, uncertain parameter, and discretization over-approximation over one sampling interval. It is not an arbitrary finite sample of a continuous disturbance set. Replacing a continuous disturbance set by finitely many extremal modes is sound only under an additional dominance property, such as disturbance-state monotonicity.}

\begin{definition}[Componentwise order]\label{def:order}
We equip the state grid with the componentwise order: $x\preceq y$ if and only if $x_j\le y_j$ for all $j\in\{1,\ldots,d\}$.
When partial orders on $\mathcal U$ and $\Delta$ are used, they are denoted by $\preceq_{\mathcal{U}}$ and $\preceq_\Delta$.
\end{definition}

\begin{assumption}[Finite grid abstraction]\label{ass:grid}
The abstract state space is the finite partially ordered set
\[
  X = \{1,\ldots,N_1\}\times\cdots\times\{1,\ldots,N_d\},
\]
with the componentwise order of \Cref{def:order}.
We write $N:=\max_j N_j$ and $|X|=\prod_j N_j$.
\end{assumption}

\begin{definition}[Lower set and antichain basis]\label{def:down}
For $S\subseteq X$, define its downward closure by
\[
  \Down S := \{q\in X : \exists\,s\in S\text{ such that } q\preceq s\}.
\]
\end{definition}
A set $K \subseteq X$ is a \emph{lower set} if $q \in K$ and $q' \preceq q$ imply $q' \in K$.
Its \emph{antichain basis} is the set of its maximal elements:
\[ \Bas(K) \coloneqq \{q \in K : \nexists\, q' \in K \setminus \{q\},\; q \preceq q'\}. \]
A lower set $K$ is fully determined by its basis: $K = \Down\Bas(K)$.

\begin{remark}[Antichain bound]\label{rem:antichain-bound}
On the grid of \Cref{ass:grid}, the maximal antichain size is bounded by $O(N^{d-1})$; every lower set therefore admits a basis of size at most $O(N^{d-1})$, one order of magnitude smaller than $|X|=O(N^d)$.
\end{remark}
\subsection{Monotone dynamics and the robust controlled pre-operator}\label{sec:prelim:pre}

\begin{definition}[Restricted robust controlled predecessor]\label{def:pre}
\deleted{For $K\subseteq X$, $\Pre(K)$ restricts the predecessor of $K$ to states in $K$.}
\added{For $A,B\subseteq X$ and $V\subseteq\mathcal U$, define}
\begin{equation}\label{eq:pre}
\added{\Pre(A,V,B)
:=
\{q\in A:\exists u\in V\ \forall\delta\in\Delta,\,
F(q,u,\delta)\in B\}.}
\end{equation}
\added{For the safety iteration, define
\[\Phi(K):=\Pre(K,\mathcal U,K),\]
which extracts from $K$ those states that can \emph{be kept} in $K$ for one step despite the disturbances.}
\end{definition}

\begin{definition}[Controlled invariant set]\label{def:ci}
A subset $K \subseteq X_S$ is \emph{controlled invariant} (with respect to system~\eqref{eq:sys} and safe set $X_S$) if for every $x \in K$, there exists $u \in \mathcal U$ such that $F(x, u, \delta) \in K$ for all $\delta \in \Delta$.
Equivalently, \replaced{$K = \Pre(K)$}{$K = \Phi(K)$}.
\end{definition}

\begin{assumption}[Lower-closed safety]\label{ass:lower-safe}
The safe state set $X_S\subseteq X$ is a lower set under $\preceq$.
\end{assumption}

\begin{definition}[State-monotonicity (SM)]\label{def:state-monotone}
The abstract transition map \eqref{eq:sys} is \emph{state-monotone} if
\[
  x_1\preceq x_2 \Longrightarrow F(x_1,u,\delta)\preceq F(x_2,u,\delta)
\]
for all $u\in\mathcal U$ and all $\delta\in\Delta$.
\end{definition}

\begin{remark}[Stronger monotonicity variants]\label{rem:monotone-levels}
Reference~\cite{saoud2024characterization} 
also equips the control set $\mathcal U$ and/or the disturbance set $\Delta$ with partial orders 
to define control-state monotone (CSM), disturbance-state monotone (DSM), and control-disturbance-state monotone (CDSM) systems.
The correctness results below require only state monotonicity (SM).
When the stronger variants hold, they allow the predecessor check to be restricted to minimal controls, maximal adversarial modes, or both.
\end{remark}

\begin{proposition}[Lower-closed predecessor]\label{prop:pre-lower}~\cite[\S~4.1, Lemma~2]{ivanova2022lazy}
Suppose \Cref{ass:grid,ass:lower-safe} hold and $F$ is state-monotone (Definition~\ref{def:state-monotone}).
Then \replaced{$\Pre(K)$}{$\Phi(K)$} is a lower set for every lower set $K\subseteq X$.
\end{proposition}
\begin{proof}
Let \replaced{$x\in\Pre(K)$}{$x\in\Phi(K)$} and $y\preceq x$.
Since $x\in K$ and $K$ is lower, $y\in K$.
Choose $u\in\mathcal U$ such that $F(x,u,\delta)\in K$ for all $\delta\in\Delta$.
By state monotonicity,
\[
  F(y,u,\delta)\preceq F(x,u,\delta)\in K \qquad \forall\,\delta\in\Delta.
\]
Lower-closedness of $K$ then gives $F(y,u,\delta)\in K$ for all $\delta$, hence \replaced{$y\in\Pre(K)$}{$y\in\Phi(K)$}.
\end{proof}

\begin{corollary}[Ordered reductions]~\cite[\S~5, Thm.~1 and Prop.~2]{saoud2024characterization}
\label{cor:extremal}
Let $K\subseteq X$ be a lower set and let \replaced{$x\in X$}{$x\in K$}.
If the stronger CSM variant of \Cref{rem:monotone-levels} holds, then it is enough to test minimal controls:
\[
  \added{x\in\Phi(K)}
  \iff
  \exists\,u\in\mathrm{Min}(\mathcal U)\ \forall\,\delta\in\Delta:\ F(x,u,\delta)\in K.
\]
If the stronger DSM variant holds, then for every fixed $u\in\mathcal U$ it is enough to test maximal adversarial modes:
\[
  \added{x\in\Phi(K)}
  \iff
  \exists\,u\in \mathcal U\ \forall\,\delta\in \mathrm{Max}({\Delta}):\ F(x,u,\delta)\in K.
\]
If the CDSM variant holds, one may therefore test only minimal controls against maximal adversarial modes.
\replaced{Let $C$ be the number of successor evaluations used in one membership test of $x\in\Pre(K)$.}{Let $C$ denote the per-state successor-evaluation count for testing $x\in\Phi(K)$.}
In the general state-monotone case, $C=|\mathcal U|\,|\Delta|$.
The ordered reductions \cite{saoud2024characterization} give
$
  C_{\mathrm{CSM}}=|\mathrm{Min}(\mathcal U)|\,|\Delta|, \quad
  C_{\mathrm{DSM}}=|\mathcal U|\,|\mathrm{Max}(\Delta)|,$
  and 
$
  C_{\mathrm{CDSM}}=|\mathrm{Min}(\mathcal U)|\,|\mathrm{Max}(\Delta)|.$ 
\end{corollary}

\begin{theorem}[Descending predecessor sequence]\label{thm:descending-seq}
Suppose \Cref{ass:grid,ass:lower-safe} hold and $F$ is state monotone.
Define
\[
  K_0 := X_S,
  \qquad
  K_{i+1} := \replaced{\Pre(K_i)}{\Phi(K_i)}, \quad i\in\N.
\]
Then each $K_i$ is a lower set, the sequence is descending, and it converges in finitely many steps to a lower set $K_*\subseteq X_S$ which is the maximal robust controlled invariant subset of $X_S$.
\end{theorem}
\begin{proof}
Proposition~\ref{prop:pre-lower} gives lower-closedness of each $K_i$ by induction.
Since $\Phi(K)\subseteq K$ for every $K$, the sequence is descending.
Because $X$ is finite, the descending sequence converges after finitely many strict inclusions.

If $K\subseteq X_S$ is a controlled invariant set,
then $K\subseteq K_0$.
Assume inductively that $K\subseteq K_i$.
Since $K\subseteq\Phi(K)$ and the map $\Phi$ is monotone with respect to set inclusion,
\[
  K \subseteq \Phi(K) \subseteq \Phi(K_i)=K_{i+1}.
\]
Hence $K\subseteq K_i$ for all $i$, and therefore $K\subseteq K_*$.
Thus $K_*$ is the maximal robust controlled invariant subset of $X_S$.
\end{proof}

\subsection{Lazy fixed-point iteration}\label{sec:prelim:lazy}

\Cref{prop:pre-lower} together with \Cref{rem:antichain-bound} means that $K_*$ and every iterate $K_i$ can be represented by its $O(N^{d-1})$-size antichain basis rather than by the full $O(N^d)$ grid. The \emph{lazy} algorithm of~\cite{saoud2019efficient}, summarized in Algorithm~\ref{alg:lazy}, exploits this by storing only the current basis $B$ and updating it.

\added{Let $R_{\mathrm{lazy}}$ denote the number of complete outer basis passes until convergence}
and $|B|=O(N^{d-1})$ denote the basis-size bound. If the membership query ``$x\in\Down B$'' scans $B$, the safety-check phase costs $O(R_{\mathrm{lazy}}CN^{2(d-1)})$, up to dimension factors. Neighbor generation performs dominance rather than successor queries and contributes $O(R_{\mathrm{lazy}}N^{2(d-1)})$ with the same scan-based representation.
\added{Since an unsafe basis element is removed and neighbors are inserted immediately, subsequent safety and neighbor checks use an evolving $B$ and are sequential across basis updates. For a removed element $b$, the explicit update in Algorithm~\ref{alg:lazy} implements $\Bas(\Down B\setminus\{b\})$ from~\cite{saoud2019efficient}.}

\begin{algorithm}[t]
\caption{Lazy fixed point iteration~\cite{saoud2019efficient}}
\label{alg:lazy}
\begin{algorithmic}[1]
\Require Lower set $X_S$, transition $F$, inputs $\mathcal U$, adversarial set $\Delta$
\State $B \gets \Bas(X_S)$
\Repeat
  \State \added{$B_{\mathrm{pass}}\gets B$}
  \State $B_{\mathrm{u}} \gets \emptyset$
  \ForAll{$b\in \replaced{B}{B_{\mathrm{pass}}}$} \Comment{safety check}
    \If{$\nexists\,u\in\mathcal U$ s.t.\ $\forall\,\delta\in\Delta:\ F(b,u,\delta)\in\Down B$}
      \State \change{$B_{\mathrm{u}} \gets B_{\mathrm{u}} \cup\{b\}$}
      \State \change{$B \gets B\setminus\{b\}$}
      \ForAll{$j\in\{1,\ldots,d\}$} \Comment{neighbor gen.}
        \State $c \gets b-e_j$
        \If{$c_j\geq 1$ and $c\notin\Down B$}
          \State $B \gets B\cup\{c\}$
        \EndIf
      \EndFor
    \EndIf
  \EndFor
\Until{$B_{\mathrm{u}}=\emptyset$}
\State \Return $\Down B$
\end{algorithmic}
\end{algorithm}

\section{Threshold-function iteration}\label{sec:method}

\Cref{sec:method:rep} introduces the threshold-function representation of lower sets. \Cref{sec:method:iter} reformulates one GFP step as a parallel collection of binary searches. \Cref{sec:method:correct} shows equivalence with Algorithm~\ref{alg:lazy} and analyzes complexity.

\subsection{The threshold-function representation}\label{sec:method:rep}

The main idea is to project a $d$-dimensional lower set $K\subseteq X$ into a $(d{-}1)$-dimensional table using a designated axis.

\begin{definition}[Designated axis, projection, threshold function]\label{def:tau}
Fix a \emph{designated axis} $d^*\in\{1,\ldots,d\}$ and set
\[ \mathcal{K} \;:=\; \prod_{j\ne d^*}\{1,\ldots,N_j\}. \]
For $x\in X$ let $\pi(x)\in\mathcal{K}$ be its projection onto the remaining $d-1$ axes, and write $x_{d^*}$ for its coordinate along $d^*$.
For a lower set $K\subseteq X$, the \emph{threshold function} is
\begin{equation}\label{eq:tau}
  \tau_K(k) \;:=\; \max\bigl\{x_{d^*} : x\in K,\ \pi(x)=k\bigr\}, \qquad k\in\mathcal{K},
\end{equation}
with the convention $\max\emptyset = 0$.
\end{definition}

Each $k\in\mathcal{K}$ indexes a one-dimensional \emph{column} of cells along the designated axis; lower-closedness forces the safe cells in that column to form a contiguous prefix $\{1,\ldots,\tau_K(k)\}$.
The function $\tau_K$ therefore stores the entire set $K$ in a single integer per column as depicted in \Cref{fig:threshold-2d}.

\begin{proposition}[Threshold characterization]\label{thm:membership}
Let $K\subseteq X$ be a lower set. Then
\[ x\in K \;\iff\; x_{d^*} \le \tau_K(\pi(x)). \]
Moreover, $\tau_K$ is \emph{antitone}: $k\preceq k'$ implies $\tau_K(k)\ge \tau_K(k')$.
\end{proposition}
\begin{proof}
If $x\in K$, the value $x_{d^*}$ is one of the values maximized in~\eqref{eq:tau} at $k=\pi(x)$, so $\tau_K(\pi(x))\ge x_{d^*}$.
Conversely, if $x_{d^*}\le\tau_K(\pi(x))$, pick $y\in K$ with $\pi(y)=\pi(x)$ and $y_{d^*}=\tau_K(\pi(x))$; then $x\preceq y$ and lower-closedness gives $x\in K$.
Antitonicity follows by the same lifting: for $k\preceq k'$, any $x\in K$ with $\pi(x)=k'$ and $x_{d^*}=v$ dominates $y$ with $\pi(y)=k,\ y_{d^*}=v$, hence $y\in K$ and $\tau_K(k)\ge v$.
\end{proof}
\begin{remark}[Memory and query complexity]\label{rem:memory}
$\tau_K$ has $|\mathcal{K}| = \prod_{j\ne d^*} N_j = O(N^{d-1})$ entries.
Storage is $O(N^{d-1})$ and every inclusion query is $O(1)$ (one projection and one comparison), compared to $O(N^{d-1})$ for a scan-based basis check.
\end{remark}

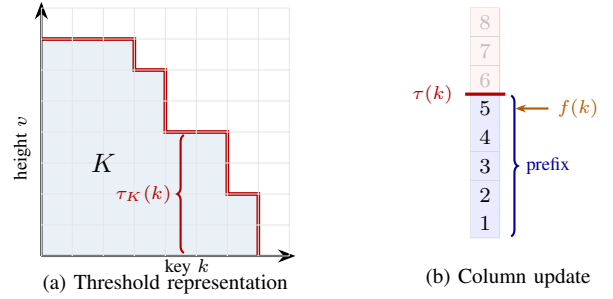
\begin{figure}[t]
\centering
\begin{minipage}[c]{0.60\columnwidth}
\centering
\begin{tikzpicture}[scale=0.41]
\fill[safeblue!12] (0,0) -- (0,7) -- (3,7) -- (3,6) -- (4,6) -- (4,4) -- (6,4) -- (6,3) -- (6,2) -- (7,2) -- (7,0) -- cycle;
\draw[red!70!black, line width=1.4pt] (0,7) -- (3,7) -- (3,6) -- (4,6) -- (4,4) -- (6,4) -- (6,3) -- (6,2) -- (7,2) -- (7,0);
\draw[-{Stealth}, thick] (0,0) -- (0,8.2);
\draw[-{Stealth}, thick] (0,0) -- (8.2,0);
\draw[very thin, gray!20] (0,0) grid (8,8);
\node[below, font=\scriptsize] at (4.7,0.2) {key $k$};
\node[rotate=90, font=\scriptsize, anchor=south] at (-0.05,3.5) {height $v$};
\draw[decorate, decoration={brace, amplitude=2.5pt}, thick, red!70!black] (4.6,0.1) -- (4.6,3.9) node[midway, left=1pt, font=\scriptsize] {$\tau_K(k)$};
\node[font=\normalsize] at (2,3) {$K$};
\node[font=\footnotesize] at (4,-0.9) {(a) Threshold representation};
\end{tikzpicture}
\end{minipage}\hfill
\begin{minipage}[c]{0.34\columnwidth}
\centering
\begin{tikzpicture}[scale=0.38]
\foreach \v in {1,...,8} {\draw[gray!30] (0.5,{\v-0.5}) rectangle (1.5,{\v+0.5});}
\foreach \v in {1,...,5} {\fill[blue!8] (0.52,{\v-0.48}) rectangle (1.48,{\v+0.48}); \node[font=\footnotesize] at (1.0,\v) {$\v$};}
\foreach \v in {6,7,8} {\fill[red!4] (0.52,{\v-0.48}) rectangle (1.48,{\v+0.48}); \node[font=\footnotesize, gray!50] at (1.0,\v) {$\v$};}
\draw[red!70!black, line width=1.2pt] (0.3,5.5) -- (1.7,5.5);
\node[left, red!70!black, font=\scriptsize] at (0.2,5.5) {$\tau(k)$};
\draw[decorate, decoration={brace, amplitude=2.5pt, mirror}, thick, blue!60!black] (1.9,0.55) -- (1.9,5.45) node[midway, right=2pt, font=\scriptsize] {prefix};
\draw[-{Stealth[length=4pt]}, thick, orange!70!black] (3.2,5.0) -- (2.05,5.0);
\node[right, orange!70!black, font=\scriptsize] at (3.2,5.0) {$f(k)$};
\node[font=\footnotesize] at (1.8,-0.9) {(b) Column update};
\end{tikzpicture}
\end{minipage}
\caption{Threshold representation and update. (a) A lower set $K$ is stored by the column height $\tau_K(k)$, so membership is the $O(1)$ test $(k,v)\in K\iff v\le\tau_K(k)$. (b) For each column, \replaced{$\Pre(K)$}{$\Phi(K)$} is a prefix whose height $f(k)$ is found by binary search; all columns are updated independently.}
\label{fig:threshold-2d}
\label{fig:column-safety}
\end{figure}


\subsection{The threshold iteration}\label{sec:method:iter}

The representation of \Cref{sec:method:rep} turns one GFP step into a set of independent binary searches along the designated axis.
We first state the structural property that enables this (\Cref{lem:column}), then derive the per-column update (\Cref{cor:safe-prefix}).

\begin{definition}[Grid column]\label{def:column}
For $k\in\mathcal{K}$, the \emph{column at $k$} is $\mathrm{Col}(k) := \{(k,v) : v\in\{1,\ldots,N_{d^*}\}\}$.
\end{definition}

\begin{lemma}[Column-wise predecessor prefix]\label{lem:column}
If $K\subseteq X$ is a lower set and
\replaced{$(k,v)\in\Pre(K)$}{$(k,v)\in\Phi(K)$}, then \replaced{$(k,v')\in\Pre(K)$}{$(k,v')\in\Phi(K)$} for every $v'\in\{1,\ldots,v\}$.
\end{lemma}
\begin{proof}
Let $u\in\mathcal U$ be a witness for \replaced{$(k,v)\in\Pre(K)$}{$(k,v)\in\Phi(K)$}, so $F((k,v),u,\delta)\in K$ for every $\delta\in\Delta$.
Since $(k,v')\preceq(k,v)$, monotonicity yields $F((k,v'),u,\delta)\preceq F((k,v),u,\delta)\in K$, and lower-closedness of $K$ gives $F((k,v'),u,\delta)\in K$.
Hence \replaced{$(k,v')\in\Pre(K)$}{$(k,v')\in\Phi(K)$}.
\end{proof}

\Cref{lem:column} says that, within any column, the set of states from which $K$ can be maintained is a contiguous prefix $\{(k,1),\ldots,(k,f(k))\}$.
This height is efficiently computable:

\begin{corollary}[Per-column threshold update]\label{cor:safe-prefix}
Let $K\subseteq X$ be lower with threshold function $\tau=\tau_K$.
For each $k\in\mathcal{K}$, define
\begin{equation}\label{eq:fk}
  \begin{aligned}
    f(k) := \max\Bigl\{&v\in\{1,\ldots,\tau(k)\}:\ \exists\,u\in\mathcal U\ \forall\,\delta\in\Delta, \\
    &F((k,v),u,\delta)\in K\Bigr\},
  \end{aligned}
\end{equation}
with $\max\emptyset:=0$.
Then
\[
  \replaced{\Pre(K)}{\Phi(K)}\cap\mathrm{Col}(k)=\{(k,1),\ldots,(k,f(k))\}.
\]
Moreover, $f(k)$ can be computed by binary search over $\{1,\ldots,\tau(k)\}$.
Each probe requires one membership query in $K$ per pair $(u,\delta)$, so the cost per column is
$
  O\!\bigl(C\,\log N\bigr).
$

\end{corollary}
\begin{proof}
By \Cref{lem:column}, \replaced{$\Pre(K)\cap\mathrm{Col}(k)$}{$\Phi(K)\cap\mathrm{Col}(k)$} is a prefix of $\mathrm{Col}(k)$ of length equal to~\eqref{eq:fk} and
the predicate in~\eqref{eq:fk} is monotone non-increasing in $v$, so a standard binary search over $\{1,\ldots,\tau(k)\}$ locates $f(k)$ in $O(\log N)$ probes.
\end{proof}

Since the update $\tau(k)\leftarrow f(k)$ is independent across columns, all columns can be processed in parallel as in \Cref{fig:column-safety}. The full iteration is summarized in Algorithm~\ref{alg:threshold}.

\begin{algorithm}[t]
\caption{Threshold-function iteration}
\label{alg:threshold}
\begin{algorithmic}
\Require Lower set $X_S$, transition $F$, inputs $\mathcal U$, adversarial set $\Delta$, designated axis $d^*$
\State $\tau \gets \tau_{X_S}$
\Repeat
  \ForAll{$k\in\mathcal{K}$} \Comment{all columns are independent}
    \State $f(k) \gets \textsc{BinSearch}(k,\tau,F, \mathcal{U}, \Delta)$ \Comment{using \eqref{eq:fk}}
  \EndFor
  \State $\tau \gets f$
\Until{$\tau$ unchanged}
\State \Return $K_\tau = \{x\in X : x_{d^*}\le \tau(\pi(x))\}$
\end{algorithmic}
\end{algorithm}

\newcommand{\rmark}[1]{\nolinebreak\hfill\tikzmark{#1}\relax}
\newcommand{\algbraceR}[3]{%
  \begin{tikzpicture}[overlay, remember picture]
    \draw[thick, decorate, decoration={brace, amplitude=3pt}]
      ([xshift=4pt, yshift=3pt]pic cs:#1) -- ([xshift=4pt, yshift=-2pt]pic cs:#2)
      node[midway, right=4pt, font=\scriptsize, align=left, inner sep=1pt] {#3};
  \end{tikzpicture}%
}

\subsection{Correctness, convergence, and complexity}\label{sec:method:correct}

We now prove that Algorithm~\ref{alg:threshold} computes the same invariant set as Algorithm~\ref{alg:lazy}.

\begin{theorem}[One-step equivalence]\label{thm:one-step}
Let $K\subseteq X$ be a lower set with threshold $\tau=\tau_K$, and let $f$ be defined by~\eqref{eq:fk}. Then, \Cref{ass:grid} and state monotonicity (\Cref{def:state-monotone}) imply:
\[ \replaced{\Pre(K)}{\Phi(K)} \;=\; \bigl\{x\in X : x_{d^*}\le f(\pi(x))\bigr\}, \
\mbox{i.e.}\ {f = \tau_{\Phi(K)}}.
\]
\end{theorem}
\begin{proof}
By \Cref{cor:safe-prefix}, $\Phi(K)\cap\mathrm{Col}(k) = \{(k,1),\ldots,(k,f(k))\}$ for every $k\in\mathcal{K}$; taking the union over $k$ gives $\Phi(K)=\{x:x_{d^*}\le f(\pi(x))\}$.
Combining this with \Cref{thm:membership} yields $f=\tau_{\Phi(K)}$.
\end{proof}

\begin{corollary}[Correctness and convergence of Algorithm~\ref{alg:threshold}]\label{cor:correct}
Let $\tau_0:=\tau_{X_S}$ and, for each $i\in\N$, let $\tau_{i+1}$ be obtained from $\tau_i$ by the update~\eqref{eq:fk}.
Define
\[
  K_i := \{x\in X : x_{d^*}\le \tau_i(\pi(x))\}.
\]
Then $K_0=X_S$, each $K_i$ is a lower set, and
\[
  K_{i+1}=\replaced{\Pre(K_i)}{\Phi(K_i)} \qquad \forall\, i\in\N.
\]
The sequence $(K_i)_{i\in\N}$ is descending and converges in finitely many steps to the maximal robust controlled invariant subset $K_*$ of $X_S$.
Equivalently, Algorithm~\ref{alg:threshold} returns $K_*$.
\end{corollary}
\begin{proof}
The identity $K_0=X_S$ follows from the definition of $\tau_{X_S}$ together with \Cref{thm:membership}.
Since $X_S$ is a lower set and \Cref{thm:one-step} maps the threshold of a lower set to the threshold of its predecessor, each $K_i$ is a lower set and satisfies $K_{i+1}=\Phi(K_i)$.
The conclusion now follows from \Cref{thm:descending-seq}.
\end{proof}

\begin{theorem}[Equivalence with the lazy basis formulation]\label{thm:equiv}
Let $K^{\mathrm{lazy}}$ denote the set returned by Algorithm~\ref{alg:lazy} and let $K^{\mathrm{thr}}$ denote the set returned by Algorithm~\ref{alg:threshold}, both started from the same lower set $X_S$.
Then $
  K^{\mathrm{lazy}} = K^{\mathrm{thr}} = K_*.
$
\end{theorem}
\begin{proof}
By \deleted{\mbox{\Cref{thm:one-step}}}\added{\Cref{cor:correct}}, Algorithm~\ref{alg:threshold} returns the maximal robust controlled invariant subset $K_*$ of $X_S$.
By the completeness result for the lazy algorithm in~\cite{saoud2019efficient}, Algorithm~\ref{alg:lazy} returns the same maximal invariant set.
\end{proof}

The threshold iteration replaces quadratic lazy boundary work with one embarrassingly parallel stage;
{see \Cref{tab:complexity}.}

{
\begin{table}[t]
\centering
\caption{Per-iteration complexity comparison.}
\label{tab:complexity}
\setlength{\tabcolsep}{2pt}
\renewcommand{\arraystretch}{1.1}
\footnotesize
\begin{tabular}{@{}p{2.5cm}cc@{}}
\toprule
 & \textbf{Lazy}~\cite{saoud2019efficient} & \textbf{Threshold (ours)} \\
\midrule
Predecessor tests & $O(CN^{2(d-1)})$ & $O(CN^{d-1}\log N)$ \\
                  & \change{\emph{sequential} over evolving $B$} & parallel over columns \\[1pt]
Basis update      & $O(N^{2(d-1)})$ & --- \\
                  & \emph{sequential} & --- \\[1pt]
\bottomrule
\end{tabular}
\end{table}
}
\begin{remark}[Dense storage and designated-axis choice]
\added{Let $M_{d^*}:=\prod_{j\ne d^*}N_j$. A threshold buffer stores $M_{d^*}$ values regardless of frontier complexity, whereas a sparse basis stores $|B|$ maximal grid points and may require substantially less memory when $|B|\ll M_{d^*}$. The threshold representation is  attractive for dense or irregular frontiers and workloads dominated by  membership queries. Memory can be minimized by choosing an axis with maximal $N_{d^*}$.}
\end{remark}

\section{Experiments}\label{sec:experiments}

We evaluate threshold iteration on three $d{=}3$ benchmarks and a $d{=}5$ benchmark from the monotone-synthesis literature.

\textbf{ACC} --- adaptive cruise control~\cite{saoud2019efficient}.
The 3D state is $(h,v_e,v_\ell)$ (headway, ego speed, lead speed), with ego torque command $u$ and adversarial mode $\delta$ that captures lead-vehicle braking and parameter uncertainty.
The ACC dynamics are
\[
\dot h=v_\ell-v_e,\qquad
\dot v_e=a_e(v_e,u,\delta),\qquad
\dot v_\ell=a_\ell(v_\ell,\delta).
\]
ACC-5D uses state $(h,v_e,v_\ell,T_e,T_\ell)$, replaces direct torque arguments by actuator states in $a_e,a_\ell$, and adds $\dot T_e=(u-T_e)/\tau_e,\ \dot T_\ell=(\delta_T-T_\ell)/\tau_\ell$.
The acceleration model is
\begin{equation}\label{eq:acc-accel}
a(v, T, \theta) = \frac{1}{M}\!\left(\frac{T}{R_w} - \alpha\right) - \frac{\beta}{M}v - \frac{\gamma}{M}v^2,
\end{equation}
where $\theta = (M, R_w, \alpha, \beta, \gamma) \in \Theta$.
The safe set is $K_{\mathrm{ACC}}=\{(h,v_e,v_\ell): h\ge d_b(v_e)\}$, where $d_b$ is the worst-case braking-distance map.
We use the orthant order
\[
x\preceq_{\mathrm{ACC}} y
\iff
h_x\ge h_y,\ v_{e,x}\le v_{e,y},\ v_{\ell,x}\ge v_{\ell,y},
\]
and, for ACC-5D, additionally $T_{e,x}\le T_{e,y}$ and $T_{\ell,x}\ge T_{\ell,y}$.
$K_{\mathrm{ACC}}$ is lower-closed whenever $d_b$ is nondecreasing.

\textbf{Turn-Ego/Turn-Oncoming.} 
Both unprotected-turn benchmarks~\cite{smith2021monotonicity} use state $(s_e,v_e,s_o)$: ego longitudinal position and speed, plus oncoming longitudinal position along the opposing lane centerline.
There is no exogenous disturbance channel, so we use a singleton adversarial set. The dynamics are
\[
\dot s_e=v_e,\qquad \dot v_e=a_e(v_e,u),\qquad \dot s_o=v_o.
\vspace{-0.1cm}
\]
Let $Z(s):=1+\mathbf 1_{\{s\ge -10\}}+\mathbf 1_{\{s\ge 10\}}$ be the shared zone map and write $Z_e:=Z(s_e),\ Z_o:=Z(s_o)$.
Both benchmarks use the same unprotected-turn geometry for opposite maneuvers: Turn-Oncoming models the ego waiting while the oncoming car passes, Turn-Ego models the ego passing first; the orthant order is flipped between them, thus the ego-priority and oncoming-priority safe sets are
\[
\begin{aligned}
K_{\mathrm{ego}}&=\{x\in X_{\mathrm{turn}}:\neg(Z_e<Z_o\lor (Z_e=2\wedge Z_o=2))\},\\
K_{\mathrm{onc}}&=\{x\in X_{\mathrm{turn}}:\neg(Z_o<Z_e\lor (Z_e=2\wedge Z_o=2))\}.
\end{aligned}
\]
They forbid simultaneous conflict-zone occupation while enforcing the chosen priority relation.
The orthant orders are
\[
\begin{aligned}
x\preceq_{\mathrm E} y
\iff s_{e,x}\ge s_{e,y},\ v_{e,x}\ge v_{e,y},\ s_{o,x}\le s_{o,y},\\
x\preceq_{\mathrm O} y
\iff s_{e,x}\le s_{e,y},\ v_{e,x}\le v_{e,y},\ s_{o,x}\ge s_{o,y}.
\end{aligned}
\]
Minimal controls are maximum torque for $\preceq_{\mathrm E}$ and minimum torque for $\preceq_{\mathrm O}$.
Since $Z$ is monotone, moving downward under $\preceq_{\mathrm E}$ advances ego and delays the oncoming phase, while downward movement under $\preceq_{\mathrm O}$ has the reverse effect; thus, priority constraints are lower-closed, and
monotonicity follows from the dynamics' sign structure~\cite{smith2021monotonicity}.

\subsection{Lazy vs. threshold comparison}

\added{For fair comparison, we also introduce a stronger Lazy+$\tau$ baseline which preserves Algorithm~\ref{alg:lazy}'s traversal and immediate-deletion schedule while maintaining the exact threshold table $\tau_{\Down B}$. Deleting an uncontrolled maximal element $b$ updates $\tau(\pi(b))\leftarrow\tau(\pi(b))-1$. A neighbor $c=b-e_j$ is maximal when $c\in K_\tau$ and every valid immediate successor $c+e_\ell\notin K_\tau$, which is checked using threshold queries. Thus, Lazy+$\tau$ changes the representation and boundary-maintenance operations while retaining the sequential Lazy schedule; the basis is materialized 
after
each pass.} \replaced{\mbox{\Cref{tab:vs-lazy}} compares basis and threshold algorithms at matching grid sizes ($|X| \approx 10^8-10^{14}$; uniform $N$ per dimension).
Lazy time $t$ includes GPU successor-table precomputation plus GFP iteration, which usually makes lazy faster due to the parallelizable precomputation.
This precomputation costs $O(C|X|)$ independent operations and is highly parallel, but makes the lazy baseline faster because lazy synthesis performs more fixed-point iterations and tests successors against the evolving basis.}{Table~\ref{tab:vs-lazy} compares scan-based Lazy, Lazy+$\tau$, and threshold GFP. The two Lazy methods use the same pass, immediate-deletion schedule, and the same precomputed successor cache. The time required to construct this cache is reported as $t_{\rm pre}$. The implementation and experiment configurations are publicly available. \footnote{\url{https://github.com/mkhaled87/pFaces-MonoSynth}}}

The threshold method uses inline dynamics: a full successor table would cost $O(C|X|)$ memory and becomes the bottleneck at larger grids, while threshold membership needs only the $O(N^{d-1})$ table and eliminates antichain maintenance. At $|X|=10^9$, the threshold iteration is already thousands of times faster; at $|X|\ge 10^{10}$ the basis method exceeds the time or memory budget, but the threshold iteration completes. \added{Let $R_{\mathrm{thr}}$ denote the number of threshold rounds. A threshold round computes one full $\Phi(K)$ update over all columns, whereas a lazy pass incrementally mutates the exposed antichain boundary. Thus, $R_{\mathrm{lazy}}$ and $R_{\mathrm{thr}}$ are algorithm-specific round counts and are not comparable.}

Experiments were run on an NVIDIA RTX~5090 GPU (32\,GB, 170 SMs) with an AMD Ryzen~9 9950X3D CPU. We implement the method in C++/pFaces~\cite{khaled2019pfaces} with OpenCL kernels.
Each key $k\in\mathcal K$ evaluates predecessor feasibility and updates $\tau(k)$ by binary search\deleted{; the host performs one synchronized table swap per iteration}.
Memory is dominated by $\tau\in\mathbb N^{|\mathcal K|}$ with $|\mathcal K|=O(N^{d-1})$.

\begin{table}[t!]
\centering
\ifshowtimingbreakdown
\else
\caption{
Lazy vs. Threshold Comparison. \\
$t_{\rm pre}$ is the successor-table generation time used in Lazy methods and is excluded from the reported speedups, which equal $t_{\rm lazy+\tau}/t_{\rm thr}$; R = iterations; timeout = 20 minutes.}
\fi
\label{tab:vs-lazy}
\setlength{\tabcolsep}{0.85pt}
\renewcommand{\arraystretch}{1.02}
\ifshowtimingbreakdown
\else
\footnotesize
\def\phasecell#1#2#3{$#1\,#2$}
\fi
\resizebox{\columnwidth}{!}{%
\begin{tabular}{@{}ccccccccc@{}}
\toprule
Case & $|X|$ & $t_{\rm pre}$ & \multicolumn{3}{c}{Lazy} & \multicolumn{3}{c}{Threshold} \\
\cmidrule(lr){4-6}\cmidrule(l){7-9}
\ifshowtimingbreakdown
\else
& & & $R_{\rm lazy}$ & $t_{\rm lazy}$ & $t_{\rm lazy+\tau}$
& $R_{\rm thr}$ & $t_{\rm thr}$ & Speedup \\
\fi
\midrule
ACC & $10^6$ & $.569\,\mathrm{ms}$ & 182 & \phasecell{63.1}{\mathrm{s}}{.656/62.5} & \phasecell{106}{\mathrm{ms}}{10.7/30.3/40.5} & 26 & \phasecell{3.18}{\mathrm{ms}}{3.14/.0370} & $33.4\times$ \\
ACC & $10^9$ & $439\,\mathrm{ms}$ & 1595 & \quad Timeout & \phasecell{126}{\mathrm{s}}{32.2/25.8/43.9} & 24 & \phasecell{16.9}{\mathrm{ms}}{16.2/.678} & $7{,}458\times$ \\
ACC & $10^{10}$ & \multicolumn{4}{c}{Timeout} & 24 & \phasecell{318}{\mathrm{ms}}{315/2.87} & - \\
ACC & $10^{12}$ & \multicolumn{4}{c}{Timeout} & 24 & \phasecell{750}{\mathrm{ms}}{684/66.5} & - \\
ACC & $10^{14}$ & \multicolumn{4}{c}{Timeout} & 24 & \phasecell{7.94}{\mathrm{s}}{6.78/1.16} & - \\
\midrule
ACC-5D & $10^6$ & $.624\,\mathrm{ms}$ & 318 & \phasecell{299}{\mathrm{s}}{2.60/297} & \phasecell{457}{\mathrm{ms}}{32.1/226/133} & 25 & \phasecell{3.34}{\mathrm{ms}}{3.23/.112} & $137\times$ \\
ACC-5D & $10^9$ & $470\,\mathrm{ms}$ & 835 & \quad Timeout & \phasecell{410}{\mathrm{s}}{89.0/160/106} & 31 & \phasecell{143}{\mathrm{ms}}{130/12.8} & $2{,}870\times$ \\
ACC-5D & $10^{11}$ & \multicolumn{4}{c}{Timeout} & 28 & \phasecell{4.94}{\mathrm{s}}{4.50/.443} & - \\
\midrule
T-Ego & $10^6$ & $1.35\,\mathrm{ms}$ & 278 & \phasecell{329}{\mathrm{s}}{1.82/327} & \phasecell{269}{\mathrm{ms}}{21.5/87.4/98.2} & 22 & \phasecell{10.8}{\mathrm{ms}}{10.8/.0400} & $24.9\times$ \\
T-Ego & $10^9$ & $1.18\,\mathrm{s}$ & 2691 & \quad Timeout & \phasecell{341}{\mathrm{s}}{58.5/87.1/125} & 26 & \phasecell{41.6}{\mathrm{ms}}{40.8/.865} & $8{,}188\times$ \\
T-Ego & $10^{10}$ & \multicolumn{4}{c}{Timeout} & 26 & \phasecell{274}{\mathrm{ms}}{270/3.92} & - \\
T-Ego & $10^{12}$ & \multicolumn{4}{c}{Timeout} & 26 & \phasecell{3.18}{\mathrm{s}}{3.10/.0763} & - \\
T-Ego & $10^{14}$ & \multicolumn{4}{c}{Timeout} & 26 & \phasecell{70.0}{\mathrm{s}}{68.4/1.53} & - \\
\midrule
T-Onc & $10^6$ & $.563\,\mathrm{ms}$ & 273 & \phasecell{118}{\mathrm{s}}{.508/117} & \phasecell{173}{\mathrm{ms}}{12.6/60.5/61.2} & 35 & \phasecell{4.22}{\mathrm{ms}}{4.17/.0400} & $41.0\times$ \\
T-Onc & $10^9$ & $458\,\mathrm{ms}$ & 2462 & \quad Timeout & \phasecell{185}{\mathrm{s}}{27.4/51.6/66.5} & 35 & \phasecell{26.4}{\mathrm{ms}}{25.6/.830} & $6{,}993\times$ \\
T-Onc & $10^{10}$ & \multicolumn{4}{c}{Timeout} & 35 & \phasecell{123}{\mathrm{ms}}{119/3.29} & - \\
T-Onc & $10^{12}$ & \multicolumn{4}{c}{Timeout} & 34 & \phasecell{960}{\mathrm{ms}}{885/75.8} & - \\
T-Onc & $10^{14}$ & \multicolumn{4}{c}{Timeout} & 35 & \phasecell{19.2}{\mathrm{s}}{17.7/1.52} & - \\
\bottomrule
\end{tabular}
}
\end{table}

\ifshowoldtable
\begin{table}[t!]
\centering
\caption{Previous Table II\\
$R_{\mathrm{lazy}},R_{\mathrm{thr}}$ = algorithm-specific rounds; $t$ = synthesis time; timeout = 20 minutes.}
\label{tab:vs-lazy-old}
\setlength{\tabcolsep}{3pt}
\renewcommand{\arraystretch}{1.05}
\begin{tabular}{@{}llrrrrrc@{}}
\toprule
System & $|X|$ & \multicolumn{2}{c}{Lazy} & \multicolumn{2}{c}{Threshold} & Speedup \\
\cmidrule(lr){3-4}\cmidrule(lr){5-6}
 & & $R_{\mathrm{lazy}}$ & $t$ & $R_{\mathrm{thr}}$ & $t$ & \\
\midrule
ACC & $10^{8}$ & 897 & 3.5\,s & 25 & 13\,ms & $269{\times}$ \\
ACC & $10^{9}$ & 2012 & 91.4\,s & 24 & 45\,ms & $2{,}031{\times}$  \\
ACC & $10^{10}$ & \multicolumn{2}{c}{Timeout} & 24 & 141\,ms & -- \\
ACC & $10^{12}$ & \multicolumn{2}{c}{Timeout} & 24 & 2470\,ms & -- \\
ACC & $10^{14}$ & \multicolumn{2}{c}{Timeout} & 24 & 43.4\,s & -- \\

\midrule
ACC-5D & $10^{8}$ & 541 & 932.7\,s & 24 & 67\,ms & $13{,}921{\times}$  \\
ACC-5D & $10^{9}$ & \multicolumn{2}{c}{Timeout}  & 26 & 449\,ms & -- \\
ACC-5D & $10^{11}$ & \multicolumn{2}{c}{Timeout}  & 28 & 16.1\,s & -- \\

\midrule
Turn-Ego & $10^{8}$   & 1269  & 49.0\,s    & 25 & 22\,ms   & $2{,}227{\times}$ \\
Turn-Ego & $10^{9}$   & 2738  & 1140\,s  & 26 & 52\,ms   & $21{,}923{\times}$ \\
Turn-Ego & $10^{10}$  & \multicolumn{2}{c}{Timeout}  & 26 & 221\,ms  & --  \\
Turn-Ego & $10^{12}$ & \multicolumn{2}{c}{Timeout}  & 26 & 4577\,ms & -- \\
Turn-Ego & $10^{14}$ & \multicolumn{2}{c}{Timeout} & 26 & 98.5\,s & -- \\
\midrule
Turn-Onc & $10^{8}$   & 1156  & 23.1\,s    & 35 & 15\,ms   & $1{,}540{\times}$ \\
Turn-Onc & $10^{9}$   & 2451  & 306\,s     & 34 & 39\,ms   & $7{,}846{\times}$ \\
Turn-Onc & $10^{10}$  & \multicolumn{2}{c}{Timeout}  & 34 & 154\,ms  & --  \\
Turn-Onc & $10^{12}$ & \multicolumn{2}{c}{Timeout}  & 34 & 3635\,ms & -- \\
Turn-Onc & $10^{14}$ & \multicolumn{2}{c}{Timeout} & 34 & 75.4\,s & -- \\
\bottomrule
\end{tabular}
\end{table}
\fi

\subsection{Real-time controller: unprotected left turn}\label{sec:exp:rt}
To demonstrate online use of the synthesized invariant sets, we employ the threshold method inside an MPPI~\cite{williams2017mppi} controller for an \emph{unprotected left turn} (\Cref{fig:rt-controller}).
The ego vehicle approaches the intersection from the west while two oncoming vehicles approach from the east.
Their velocities are sensed online.
The ego first waits for $\mathrm{onc}_1$ to clear the conflict zone, then commits to the turn before $\mathrm{onc}_2$ closes the available gap.
Each phase is formulated as a 3D robust-safety problem in $(s_e,v_e,s_{\mathrm{onc}})$, and the corresponding invariant set is re-synthesized when the sensed oncoming speed changes sufficiently.
Let $K_{\mathrm{w}}\subseteq\R^3$ and $K_{\mathrm{g}}\subseteq\R^3$ be the two phase sets.
The two-oncoming safety condition is represented by the intersection of these two 3D constraints in lifted coordinates:
\[
K_{2}
=\{(s_e,v_e,s_1,s_2):\ 
(s_e,v_e,s_1)\in K_{\mathrm{w}},\\
 (s_e,v_e,s_2)\in K_{\mathrm{g}}\}.
\]
This lifted set is used only for online rollout membership queries; we do not claim that the full 4D system is monotone or that $K_2$ is a 4D controlled-invariant set.
The reduction is sound for simultaneously checking the two phase certificates under the modeled independent oncoming evolution.
A formal 4D invariance claim would additionally require a common witness input for both component constraints, which need not exist because the wait and go constructions impose opposite orthant orders on $(s_e,v_e)$.

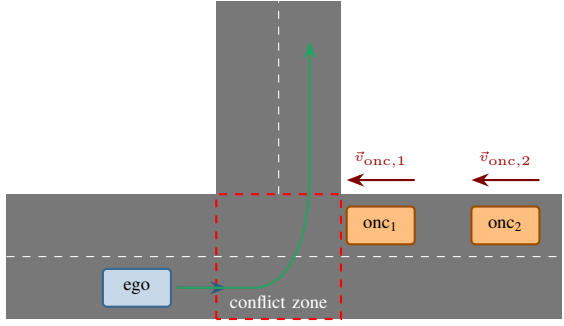
\begin{figure}[t]
\centering
\begin{tikzpicture}[scale=0.75,
  >=Stealth,
  font=\scriptsize,
  car/.style={draw, thick, rounded corners=1.5pt, minimum width=9mm, minimum height=5mm, inner sep=0pt},
  egocar/.style={car, fill=safeblue!30, draw=safeblue!80!black},
  onccar/.style={car, fill=orange!50, draw=orange!60!black},
]
\definecolor{asphalt}{RGB}{120,120,120}

\begin{scope}[xshift=0cm]
\fill[asphalt] (-1.0,-1.1) rectangle (8.8,1.1);
\fill[asphalt] (2.7,1.1) rectangle (4.9,4.5);
\draw[dashed, white!60] (3.8,1.1) -- (3.8,4.5);
\fill[asphalt, draw=red!100, dashed, thick] (2.7,-1.1) rectangle (4.9,1.1);
\node[red!0] at (3.8,-0.8) {conflict zone};
\draw[dashed, white!60] (-1.0,0.0) -- (8.8,0.0);

\node[egocar] (ego) at (1.3,-0.55) {ego};
\draw[->, safeblue!70!black, thick] (2.0,-0.55) -- (2.9,-0.55);
\node[onccar] (onc1) at (5.6,0.55) {onc\textsubscript{1}};
\draw[->, red!50!black, thick] (6.2,1.35) -- (5.0,1.35) node[midway, above=1pt, font=\tiny, red!60!black] {$\vec v_{\mathrm{onc},1}$};
\node[onccar] (onc2) at (7.8,0.55) {onc\textsubscript{2}};
\draw[->, red!50!black, thick] (8.4,1.35) -- (7.2,1.35) node[midway, above=1pt, font=\tiny, red!60!black] {$\vec v_{\mathrm{onc},2}$};
\draw[accentgreen, thick, -{Stealth}, line width=0.8pt]
  (2.0,-0.55) -- (3.35,-0.55)             
  .. controls (3.95,-0.55) and (4.35,0.20) .. (4.35,1.5)  
  -- (4.35,3.8);                          
\end{scope}

\end{tikzpicture}
\caption{Real-time controller for an unprotected left turn. \replaced{The threshold table acts as a rollout safety filter.}{An MPPI controller samples rollouts and queries the threshold table once per rollout step to evaluate the set-membership penalty.} The resulting optimized control input is applied. When 
oncoming velocities change, 
the invariant set is resynthesized.}
\label{fig:rt-controller}
\end{figure}


The invariant-set algorithm returns $K_*$, not an explicit policy.
In closed loop, we use $K_*$ \deleted{as a real-time certificate} inside an optimizer: the threshold table gives an $O(1)$ membership query that is easy to insert into a sampling-based MPPI cost\deleted{ or safety filter}.
MPPI is a natural choice here because it handles nonlinear dynamics and nonsmooth table penalties without requiring gradients of the threshold representation; first-order MPC methods would need smoothing of the same discontinuous membership test.

Let $T_s$ denote the controller sampling period, $H_p$ the prediction horizon, $H_c$ the control horizon, $N_r$ the number of rollouts per MPPI iteration, $N_{\mathrm{it}}$ the number of MPPI updates per control step, $\sigma_u$ the standard deviation of the Gaussian input perturbations, and $[u_{\min},u_{\max}]$ the admissible torque interval.
Their numerical values are reported in \Cref{tab:rt}.
The stage cost is
$
\ell(x_k,u_k)=w_p(s_{e,k}-g)^2+w_s\mathbf 1_{\{x_k\notin K_{2}\}}+w_j\|u_k-u_{k-1}\|^2
$.
Each table is computed for the oncoming-speed parameters and interval vehicle parameters, and re-synthesis is triggered when the sensed velocity change exceeds the threshold $\varepsilon_{\mathrm{res}}$ listed in \Cref{tab:rt}.
The  \replaced{safety filter}{safety-informed MPPI} performs an $O(1)$ threshold lookup per rollout step; states outside $K_{\mathrm{2}}$ accumulate a penalty $w_s$\added{, rather than being prohibited by a hard constraint, and the implemented controller provides no formal closed-loop safety guarantee. Such a guarantee would require a hard one-step invariant constraint or an invariant-preserving backup policy, 
which is 
left for future work}.
The threshold lookup and MPPI control computation times are negligible relative to the re-synthesis times. 

\begin{table}[t]
\centering
\caption{Real-time unprotected left turn: configuration and timings.}
\label{tab:rt}
\setlength{\tabcolsep}{3pt}
\renewcommand{\arraystretch}{1.05}
\begin{tabular}{@{}p{5.0cm}p{3.0cm}@{}}
\toprule
Quantity & Value \\
\midrule
Grid cells $|X|$ & $10^{10} = 10^4 \times 10^2 \times 10^4$ \\
Sampling period $T_s$ & $0.1$\,s \\
Prediction and control horizon $H_p, H_c$ & $20$ steps \\
MPPI iterations / step $N_{\mathrm{it}}$ & $10$ \\
Rollouts / iteration $N_r$ & $512$ \\
Perturbation standard deviation $\sigma_u$ & $1000$\,Nm \\
Cost weights $(w_p,w_s,w_j)$ & $(1,10^4,10^{-2})$ \\
Resynthesis threshold $\varepsilon_{\mathrm{res}}$ & $0.2$\,m/s \\
Wait and go synthesis $t^{\mathrm{wait}}, t^{\mathrm{go}}$ & $77 \pm 15, 62 \pm 7\,ms$ \\
Average MPPI control time & $2$\,ms \\
\bottomrule
\end{tabular}
\end{table}

\section{Conclusion}\label{sec:conclusion}
We introduced a threshold-function reformulation of greatest-fixed-point invariant-set synthesis for monotone systems.
The algorithm is embarrassingly parallel and orders of magnitude faster than lazy fixed-point iteration.
Experiments on \deleted{ACC and unprotected-left-turn} benchmarks confirm speedups \replaced{up to  $21{,}924\times$}{on the order of thousands} over the  lazy baseline.
The threshold iteration reaches grids exceeding $10^{9}$ cells around 100\,ms, enabling real-time re-synthesis for closed-loop control.
\deleted{We demonstrated this by embedding it in an MPPI controller that re-synthesizes the invariant set on-the-fly as oncoming vehicle speeds are sensed.}
However, the method requires monotonicity under some orthant order and lower-closedness of the safety set.
The threshold table stores only the set and gives no specific control policy.
Grid resolution is bounded by GPU memory ($O(N^{d-1})$ integers).
\bibliography{main}
\end{document}